\documentclass[11pt]{article}

\usepackage{wrapfig}
\usepackage{graphicx}
\usepackage{mathpazo}
\usepackage{amsmath}               
\usepackage{amsfonts}              
\usepackage{amsthm}                
\usepackage{amssymb}
\usepackage{todonotes}
\usepackage[pagebackref]{hyperref}
\usepackage[margin=1in]{geometry}
\usepackage{todonotes}

\newtheorem{thm}{Theorem}[section]
\newtheorem{lem}[thm]{Lemma}

\renewcommand{\leq}{\leqslant}

\renewcommand{\geq}{\geqslant}

\DeclareMathOperator{\E}{\mathbb{E}}

\DeclareMathOperator{\calI}{\mathcal{I}}

\newcommand{\caldb}{\mathcal{D}^{\mathsf{B}}}
\newcommand{\caldg}{\mathcal{D}^{\mathsf{G}}}
\newcommand{\Bino}{\mathsf{Binomial}}

\newcommand{\RR}{\mathbb{R}}      

\begin{document}

\title{{\bf Understanding the Correlation Gap for Matchings}}

\author{
Guru Guruganesh\thanks{{\tt euiwoonl@cs.cmu.edu} }}

\author{
Guru Guruganesh\thanks{Supported in part by Anupam Gupta\rq{}s NSF awards CCF-1319811 and CCF-1536002. {\tt ggurugan@cs.cmu.edu}} \and 
Euiwoong Lee\thanks{Supported by a Samsung Fellowship and Venkat Guruswami\rq{}s NSF CCF-1526092. {\tt euiwoonl@cs.cmu.edu} }}

\date{Computer Science Department \\ Carnegie Mellon University \\ Pittsburgh, PA 15213.}

\maketitle

\begin{abstract}
    Given a set of vertices $V$ with $|V| = n$, a weight vector $w \in (\RR^+ \cup \{ 0 \})^{\binom{V}{2}}$, and a probability vector $x \in [0, 1]^{\binom{V}{2}}$ in the matching polytope, we study the quantity 
\[
\frac{\E_{G}[ \nu_w(G)]}{\sum_{(u, v) \in \binom{V}{2}} w_{u, v} x_{u, v}}
\]
where $G$ is a random graph where each edge $e$ with weight $w_e$ appears with probability $x_e$ independently, and let $\nu_w(G)$ denotes the weight of the maximum matching of $G$. 
This quantity is closely related to correlation gap and contention resolution schemes, which are important tools in the design of approximation  algorithms, algorithmic game theory, and stochastic optimization. 

    We provide lower bounds for the above quantity for general and bipartite graphs, and for weighted and unweighted settings. 
    The best known upper bound is $0.54$ by Karp and Sipser, and the best lower bound is $0.4$. We show that it is at least $0.47$ for unweighted bipartite graphs, at least $0.45$ for weighted bipartite graphs, and at least $0.43$ for weighted general graphs. 
To achieve our results, we construct local distribution schemes on the dual  which may be of independent interest. 
\end{abstract}

\section {Introduction}
We study the size (weight) of the maximum matching of a random graph sampled from various random graph models. 
Let $V$ be the set of vertices with $|V| = n$.
Given the probability vector $x \in [0, 1]^{\binom{V}{2}}$ and the weight vector
$w \in (\mathbb{R}^{+} \cup \{ 0 \})^{\binom{V}{2}}$, let $\caldg_{n, w, x}$ be the distribution of random graphs with $n$ vertices
such that each pair $e \in \binom{V}{2}$ becomes an edge with probability $x_e$ independently. If it becomes an edge, its weight is $w_e$. 
For bipartite graphs, let $V_1$ and $V_2$ be the set of left and right vertices with $|V_1| = |V_2| = n$.
Given the probability vector $x \in [0, 1]^{V_1 \times V_2}$ and the weight vector
$w \in (\mathbb{R}^{+} \cup \{ 0 \})^{V_1 \times V_2}$, let $\caldb_{n, w, x}$ be the distribution of random bipartite graphs with $2n$ vertices
such that each pair $e \in V_1 \times V_2$ becomes an edge with probability $x_e$ independently. If it becomes an edge, its weight is $w_e$. 
We use $\caldb_{n, x}$ (resp. $\caldg_{n, x}$) for the unweighted case ($w = (1, 1, \dots, 1)$). 

We focus on the case when the probability vector $x$ is in the matching polytope of the complete (bipartite) graph. 
Recall that for bipartite graphs, $x \in [0, 1]^{V_1 \times V_2}$ is in the matching polytope if each $v \in V_1 \cup V_2$ satisfies $\sum_u x_{u, v} \leq 1$.
For general graphs, $x \in [0, 1]^{\binom{V}{2}}$ is in the matching polytope if each $v \in V$ satisfies $\sum_u x_{u, v} \leq 1$ 
and each odd set $S \subseteq V$ satisfies $\sum_{\{ u, v \} \subseteq S} x_{u, v} \leq \lfloor (|S| - 1) / 2 \rfloor$.\footnote{Our result for general graphs, Theorem~\ref{thm:general_weighted} holds even when $x$ satisfies the first type of constraints.} 

Given a weighted graph $G$, let $\nu_w(G)$ be the weight of the maximum weight matching of $G$. 
If $G$ is unweighted, $\nu(G)$ denotes the cardinality of the maximum matching of $G$. 
For any $x \in [0, 1]^{\binom{V}{2}}$ and $w \in (\mathbb{R}^{+} \cup \{ 0 \})^{\binom{V}{2}}$, 
we have $\E_{G \sim \caldg_{n, w, x}}[\nu_w(G)] \leq \sum_{(u, v) \in \binom{V}{2}} w_{u, v} x_{u, v}$,
simply because the probability that $(u, v)$ is included in the maximum matching is at most $x_{u, v}$. 
The analogous statement also holds for bipartite graphs. 

If $x$ is in the matching polytope\footnote{If $x$ is not in the matching polytope, 
    one can construct examples where  $\kappa = \Omega(n)$.}, we can prove that 
$\E_{G} [\nu_w(G)] \geq \kappa \cdot \sum w_{u, v} x_{u, v}$ 
for some constant $0 < \kappa < 1$ .
For the general graph model, $\kappa$ is known to be at least $(1- 1/e)^2 \sim 0.40$ for every $w$~\cite{CGM13}. 
For the bipartite graph model, $\kappa$ is known to be at least $0.4$ for every $w$~\cite{CVZ14}.
Karp and Sipser~\cite{KS81} showed an upper bound of $0.54$ for both bipartite and general graphs, 
by demonstrating it for the unweighted models where every edge appears with equal probability.
Our main results are the following improved lower bounds on $\kappa$. 
Our first theorem concerns the unweighted bipartite model. 
\begin{thm}
Let $|V_1| = |V_2| = n$ and $x \in [0, 1]^{V_1 \times V_2}$ be in the matching polytope of the complete bipartite graph on $V_1 \cup V_2$. Then 
\begin{equation}
\frac{\E_{G \sim \caldb_{n, x}} [ \nu(G) ] }{ \sum_{(u, v) \in V_1 \times V_2} x_{u, v}} \geq 0.476. 
\end{equation}
\label{thm:bipartite_unweighted}
\end{thm}
We also obtain a slightly weaker result on the weighted bipartite model. 
\begin{thm}
Let $|V_1| = |V_2| = n$ and $x \in [0, 1]^{V_1 \times V_2}$ be in the matching polytope of the complete bipartite graph on $V_1 \cup V_2$. Then 
for any $w \in (\RR^+ \cup \{ 0 \})^{V_1 \times V_2}$, 
\[
\frac{\E_{G  \sim \caldb_{n, w, x}} [\nu_w(G)]}{\sum_{(u, v) \in V_1 \times V_2} w_{u, v} x_{u, v}} \geq \bigg( 1-\frac{3}{2e}   \bigg) \geq 0.4481. 
\]
\label{thm:bipartite_weighted}
\end{thm}
Finally, we prove an improved bound on the weighted general graph model. 
\begin{thm}
Let $|V| = n$ and $x \in [0, 1]^{\binom{V}{2}}$ be in the matching polytope of the complete graph on $V_1 \cup V_2$. Then 
for any $w \in (\RR^+ \cup \{ 0 \})^{\binom{V}{2}}$, 
\[
\frac{\E_{G \sim \caldg_{n, w, x}} [\nu_w(G)]}{\sum_{(u, v) \in \binom{V}{2}} w_{u, v} x_{u, v}}  \geq \frac{e^2 - 1}{2e^2}  \geq 0.4323.\]
\label{thm:general_weighted}
\end{thm}

\subsection{Applications and Related Work.}
\paragraph*{Contention Resolution Schemes and Correlation Gap.}
 
Our work is inspired by and related to the rounding algorithms studied
in approximation algorithms.  Given a downward-closed family $\calI \subseteq
2^E$ defined over a ground-set $E$ and a submodular function $f : 2^E \rightarrow \RR^+$, Chekuri et
al.~\cite{CVZ14} considered the problem of finding $\max_{S \in \calI} f(S)$
and introduced {\em contention resolution schemes} (CR schemes) to obtain
improved approximation algorithms for numerous problems. Let $P_{\calI}$ be 
the convex combination of all incidence vectors $\{ 1_{S} \}_{S \in \calI}$. A
$c$-CR scheme $\pi$ for $x \in P_{\calI}$ is a procedure that, when
$R$ is a random subset of $E$ with $e \in R$ independently with probability
$x_e$, returns $\pi(R) \subseteq R$ such that $\pi(R) \in \calI$ with
probability $1$ and $\Pr[e \in \pi(R)] \geq c$ for all $e \in E$.  

To construct a CR scheme, they introduced the notion of {\em correlation gap of a polytope}, inspired by~\cite{ADSY12}.\footnote{\cite{ADSY12} defined the correlation gap of a set function $f : 2^E \to \RR^+$. Our results apply to this definition too when $f$ denotes the weight of the maximum matching. } 
Formally, the correlation gap of $\calI$ is defined as
\begin{equation}
    \label{def:cor}
    \kappa(\calI) := \inf_{x \in P_{\calI}, \, w \geq 0} \frac{\E_{R \sim {\mathcal D}_x}[\max_{S \subseteq R, S \in \calI} \sum_{e \in S} w_e ]}{\sum_{e \in E} x_e w_e },
\end{equation}
where ${\mathcal D}_x$ is the distribution where each element $e$ appears in $R$ with probability $x_e$ independently. 
It is easy to see that the existence of $c$-CR scheme for all $x \in P_{\mathcal{I}}$ implies $\kappa(\calI) \geq c$. 
Chekuri et al.~\cite{CVZ14} proved the converse that every $x \in P_{\mathcal{I}}$ admits a $\kappa(\calI)$-CR scheme.  

By setting $E$ to be the set of all possible edges of a complete (bipartite) graph, 
and $\calI$ to be the set of all matchings of a complete graph, 
our Theorem~\ref{thm:bipartite_weighted} and Theorem~\ref{thm:general_weighted} for weighted bipartite graphs and weighted general graphs 
imply that there exist $0.4481$-CR scheme and $0.4323$-CR scheme for bipartite matching polytopes and general matching polytopes respectively.
Note that these lower bounds hold when $E'$ is the set of edges and $\calI'$ is a matching polytope of an arbitrary graph $G'$, since 
\begin{align*}
    \kappa(\calI) &= \inf_{x \in P_{\calI}, \, w \geq 0} \frac{\E_{R \sim {\mathcal D}_x}[\max_{S \subseteq R, S \in \calI} \sum_{e \in S} w_e ]}{\sum_{e \in E} x_e w_e } \\
    &\leq \inf_{x|_{E'} \in P_{\calI'}, \, w|_{E'} = 0} \frac{\E_{R \sim {\mathcal D}_x}[\max_{S \subseteq R, S \in \calI} \sum_{e \in S} w_e ]}{\sum_{e \in E} x_e w_e }
    = \kappa(\calI').
\end{align*}

\paragraph*{Maximum Matching of Random Graphs.}

The study of maximum matchings in random graphs has a long history. It was pioneered
by the work of Erd{\H{o}}s and R{\'e}nyi~\cite{ER66, ER68}, where they proved that a random
graph $G_{n, p}$ has a perfect matching with high probability when 
$p = \Omega (\frac{\ln n}{n} )$. The case for sparse graphs was investigated
by Karp and Sipser~\cite{KS81} who gave an accurate estimate of $\nu(G)$ for  $G_{n, p}$ where
$p = \frac{c}{n - 1}$ for some constant $c > 0$.

After these two pioneering results, subsequent work has addressed two aspects.
The Karp-Sipser algorithm is a simple randomized greedy algorithm, and the first line of works
extend the range of models where this algorithm (or its variants) 
returns an almost maximum matching.  Aronson et al.~\cite{AFP98} and Chebolu et
al.~\cite{CFM10} augmented the Karp-Sipser algorithm to achieve tighter results
in the standard $G_{n, p}$ model. Bohman and Frieze~\cite{BF11} considered a
new model where a graph is drawn uniformly at random from the collection of
graphs with a fixed degree sequence and gave a sufficient condition where the
Karp-Sipser algorithm finds an almost perfect matching. 

The second line of work is based on the following observation: the standard 
$G_{n, p}$ model, $p = \Omega(\frac{\ln n}{n})$ is required to have a perfect
matching, because otherwise there will be an isolated vertex.  This naturally led
to the question of finding a natural and {\em sparser} random graph model with
a perfect matching. The considered models include a random regular graph, and a
$G_{n, p}$ with prescribed minimal degree. We refer the reader to the work of
Frieze and Pittel~\cite{FP04} and
Frieze~\cite{Frieze05} and references therein. 

\subsection{Organization}
\label{subsec:techniques}
Our main technical contribution is lower bounding correlation gaps via local distribution schemes for dual variables, 
which are used to prove Theorem~\ref{thm:bipartite_unweighted} and Theorem~\ref{thm:bipartite_weighted} 
for unweighted and weighted bipartite graphs. 
We present this framework in Section~\ref{subsec:techniques} and 
prove our bounds for unweighted bipartite graphs (Section~\ref{sec:bipartite_unweighted}) and 
weighted bipartite graphs (Section~\ref{sec:bipartite_weighted}).
Our result for weighted general graphs is presented in Section~\ref{sec:general_weighted}. 

\section{Techinques for Bipartite Graphs}
\label{subsec:techniques}
Let $V = V_1 \cup V_2$ be the set of vertices with $|V_1| = |V_2| = n$, $E := V_1 \times V_2$. 
Fix $w \in ( \RR^+ \cup \{ 0 \} )^E$ and $x \in [0, 1]^E$ in the bipartite matching polytope of $(V, E)$.

Our proofs for Theorem~\ref{thm:bipartite_unweighted} and~\ref{thm:bipartite_weighted} for bipartite graphs follow the following general framework. Let $G = (V, E(G))$ be a sampled from the distribution where each potential edge $e \in E$ appears with probability $x_e$ independently (recall that $E = V_1 \times V_2$ is the set of all potential edges and $E(G)$ is the edges of one sample $G$). 
Let $y(G) \in ( \RR^+ \cup \{ 0 \} )^V$ be an optimal {\em fractional vertex cover} such that for every $e = (u, v) \in E(G)$, $y_u(G) + y_v(G) \geq w_e$. 
By K\"{o}nig-Egerv\'{a}ry theorem, $\| y (G) \|_1 = \nu(G)$. 

Given $G$, consider the situation where initially each vertex $v$ has {\em mass} $y_v(G)$, and each potential edge has mass $y_e(G) = 0$ (we slightly abuse notation and consider $y(G) \in (\RR^+ \cup \{0 \})^{V \cup E}$). We construct {\em local distribution schemes} $F_G : (V \cup E) \times (V \cup E) \to \RR$ where $F_G(a, b)$ indicates the amount of mass sent from $a$ to $b$. We require that $F_G(a, a) = 0$, but we allow $F_G(a, b) \neq -F_G(b, a)$ for $a \neq b$ (the net flow from $a$ to $b$ in this case is $F_G(a, b) - F_G(b, a)$). 
Let $t(G) \in \RR^{V \cup E}$ denote the mass of each vertex and edge after the distribution.
\[
t_a(G) := y_a(G) + \sum_{b \in V \cup E} F_G(b, a) 
- \sum_{b \in V \cup E} F_G(a, b). 
\]
We choose $F_G$ so that it ensures $t_v(G) \geq 0$ for every $v \in V$. This implies 
\[
\sum_{e \in E} t_e(G) \leq 
\sum_{a \in V \cup E} t_a(G) = 
\sum_{a \in V \cup E} y_a(G) =
\sum_{v \in V} y_v(G) = \nu(G). 
\]
Therefore, if we prove that for each potential edge $e \in E$
\begin{equation}
\E_G [t_e(G)] \geq \alpha \cdot w_e x_e, 
\label{eq:local}
\end{equation}
for some $\alpha > 0$, it implies that 
\[
\E_G[\nu(G)] \geq \alpha \cdot \sum_{e \in E} \E_G[ t_e(G) ] \geq \alpha \cdot \sum_{e \in E} w_e x_e.
\]
For weighted and unweighted cases, we construct different local distribution schemes $\{ F_G \}_G$ that prove~\eqref{eq:local} with different values of $\alpha$. 

\paragraph*{Weighted Bipartite Graphs.} 
Given a sample $G = (V, E(G))$ and a fractional vertex cover $y \in (\RR^+ \cup \{ 0 \})^{V }$, our $F_G(v, e) = y_v(G) / \deg_G(v)$ if $e \in E(G)$ is an edge incident on $v \in V$, and $0$ otherwise. 
Intuitively, each vertex $v$ {\em distributes} its mass $y_v(G)$ evenly to its incident edges in $G$.
This clearly satisfies $t_v(G) \geq 0$ for every $v \in V$, and for each $e = (u, v) \in E$, we use the following approximation: 
\begin{align*}
\E_G[t_e(G)] & =
\Pr[e \in G] \cdot \E_G \bigg[t_e(G) | e \in G\bigg] \\
& = x_e \E_G\bigg[ \frac{y_u(G)}{\deg_G(u)} + \frac{y_v(G)}{\deg_G(v)} | e \in G\bigg] \\
& \geq x_e \E_G\bigg[ (y_u(G) + y_v(G)) \frac{1}{\max(\deg_G(u), \deg_G(v))} | e \in G\bigg] \\
& \geq x_e w_e \E_G\bigg[ \frac{1}{\max(\deg_G(u), \deg_G(v))} | e \in G\bigg].
\end{align*}
Therefore, to prove Theorem~\ref{thm:bipartite_weighted}, it suffices to prove that for every potential edge $e \in E$,
\[
\E_{G \sim \caldb_{n, w, x}}\bigg[ \frac{1}{\max(\deg_G(u), \deg_G(v))} | e \in G\bigg] \geq 0.4481, 
\]
when $G$ is sampled from $\caldb_{n,w,x}$ with $x$ in the matching polytope. 
Experimentally trying several extreme cases indicates that 
the worst case for $e = (u, v) \in E$ happens when $x_e = \varepsilon$ for very small $\varepsilon$, 
$u$ has only one other edge $e_u$ with $x_{e_u} = 1 - \varepsilon$, and $v$ is incident on $n - 1$ edges $e_{v_1}, \dots, e_{v_{n-1}}$ with $x_{e_{v_i}} = \frac{1 - \varepsilon}{n - 1}$. 
As $\varepsilon$ approaches to $0$, 
$\E_G[ \frac{1}{\max(\deg_G(u), \deg_G(v))} | e \in G]$ converges to 
$\E[\frac{1}{1 + Y_U}]$ as $n$ grows, where $Y_U$ is drawn from a binomial distribution $B(n-1, \frac{1}{n - 1})$. 
Section~\ref{sec:bipartite_weighted} formally proves that this is indeed the worst case. 

\paragraph*{Unweighted Bipartite Graphs.} 


One simple but important observation is that in the above example where 
$\E_G[t_e(G)] \approx \E[\frac{1}{1 + Y_U}] x_e$, $e$ is an edge with very small $x_e = \varepsilon$, and it is adjacent to a large edge $e_u$ with $x_{e_u} = 1 - \varepsilon$. 
From the persepctive of $x_{e_u}$, the expected number of adjacent edges is at most $2 \varepsilon$, so $\E_G[t_{e_u}(G)] \approx x_{e_u} \approx 1$. 
Since $e_u$ gets much more than what it needs ($\E[t_{e_u}] \geq 0.476$ suffices to prove Theorem~\ref{thm:bipartite_unweighted}), it is natural to take some value from $t_{e_u}(G)$ to increase $t_{e}(G)$. 

Formally, given $G = (V, E(G))$, our new local distribution scheme $F_G : (V \cup E) \times (V \cup E) \to \RR$ is defined as follows.
Let $c$ be an universal constant that will be determined later. 
\begin{equation}
\label{eq:F}
F_G(a, b) = 
\begin{cases}
\frac{y_a(G)}{\deg_G(a)} \quad & \mbox{ if } a \in V, b \in E(G), a \in b \\
c x_a^2 x_b & \mbox{ if } a \neq b \in E, a \cap b \neq \emptyset \\
0 & \mbox{ otherwise.} 
\end{cases}
\end{equation}

Intuitively, on top of the old local distribution scheme for weighted graphs, each edge $e$ pays $c x_e^2 x_f$ to every adjacent edge $f$ with probability $1$ (this quantity does not depend on $G$).
Because this term quadratically depends on the $x$ value of the sender, this payment penalizes edges with large $x$ values to help edges with small $x$ values.
For a fixed edge $e = (u, v) \in E$ with fixed $x_e = \varepsilon$, 
Theorem~\ref{thm:mass2} shows that the worst case is when both $u$ and $v$ have $n-1$ other edges of whose $x$ values are equal to $\frac{1 - \varepsilon}{n}$. 
Finally, Lemma~\ref{lem:finish} shows that $\E[t_e] \geq 0.476$ for every $\varepsilon \in (0, 1]$, proving Theorem~\ref{thm:bipartite_unweighted}.

\section {Unweighted Bipartite Graphs}
\label{sec:bipartite_unweighted}
We prove Theorem~\ref{thm:bipartite_unweighted} for unweighted bipartite graphs. Given $G = (V, E(G))$, consider the local distribution scheme $F_G : (V \cup E) \times (V \cup E) \to \RR$ given in~\eqref{eq:F}. 
This implies that the mass after this new distribution scheme for an edge $e = (u, v)$ is given by 
\[
t_e(G) = \alpha_e(G) + \sum_{f \in E \setminus \{ e \}: f \ni u} c(x_e x_f^2 - x_e^2 x_f) +
\sum_{g \in E  \setminus \{ e \}: g \ni v} c(x_g^2 x_e - x_e^2 x_g),
\]
where $\alpha_e(G) := y_u(G) / \deg_G(u) + y_v(G) / \deg_G(v)$ denotes the mass after the old distribution scheme used for weighted bipartite graphs. 
We define $\beta_e(x)$ to be the following.
\begin{align*}
\beta_e(x) := \, &
\E_{G \sim \caldb_{n,x} }[t_e(G)]  \\
= \, &  \E_{G \sim \caldb_{n,x}} [\alpha_e(G)] + \sum_{f \in E \setminus \{ e \}: f \ni u} c(x_e x_f^2 - x_e^2 x_f) +
\sum_{g \in E  \setminus \{ e \}: g \ni v} c(x_g^2 x_e - x_e^2 x_g)
\end{align*}

To prove Theorem~\ref{thm:bipartite_unweighted}, it suffices to prove that $\beta_e(x) \geq 0.476 x_e$ for each $e$. 
Fix $e = (u, v)$. Let $e_{u_1}, \dots, e_{u_{n-1}}$ be $n-1$ other edges incident on $u$ and $e_{v_1}, \dots, e_{v_{n-1}}$ be $n - 1$ other edges incident on $v$. 
$\E_{G \sim \caldb_{n,x}} [\alpha_e(G)]$ is lower bounded by 
$x_e \E_G[ \frac{1}{\max(\deg_G(u), \deg_G(v))} | e \in G]$ as before. 
Define $F(x_0, y_1, \dots, y_{n - 1}, z_1, \dots, z_{n - 1})$ by 
\begin{align*}
F(x_0, y_1, \dots, y_{n - 1}, z_1, \dots, z_{n - 1}) &:= 
x_0 \E[\frac{1}{1 + \max(Y, Z)}] + \sum_{i = 1}^{n - 1} c (x_0 y_i^2 - x_0^2 y_i) \\ 
    &\quad + \sum_{i = 1}^{n - 1} c (x_0 z_i^2 - x_0^2 z_i),
\end{align*}
where $Y := Y_1 + \dots + Y_{n - 1}$ and $Z := Z_1 + \dots + Z_{n - 1}$ and each $Y_i$ (resp. $Z_i$) is an independent Bernoulli random variable with $\E[Y_i] = y_i$ (resp. $\E[Z_i] = z_i$). 
By construction, $\beta_e(x) \geq F(x_e, x_{e_{u_1}}, \dots, x_{e_{u_{n - 1}}}, x_{e_{v_1}}, \dots, x_{e_{v_{n - 1}}})$. Given fixed $\sum_{i=1}^{n-1} x_{e_{u_i}}$ and $\sum_{i=1}^{n-1} x_{v_{u_i}}$, the following theorem shows that $F$ is minimized when $x_{e_{u_1}} = \dots = x_{e_{u_{n - 1}}}$ and 
$x_{e_{v_1}} = \dots = x_{e_{v_{n - 1}}}$. 

\begin{thm}
\label{thm:mass2}
For $x_0, y_1, \dots, y_m, z_1, \dots, z_m \in [0, 1]$ where $y_s := \sum_{i=1}^m y_i \leq 1 - x_0$ and $z_s := \sum_{i=1}^m z_i \leq 1 - x_0$, 
\begin{align*}
F(x_0, y_1, \dots, y_m, z_1, \dots, z_m) 
\geq F(x_0, \frac{y_s}{m}, \dots, \frac{y_s}{m}, \frac{z_s}{m}, \dots, \frac{z_s}{m}).
\end{align*}
\end{thm}
\begin{proof}
Without loss of generality, assume $y_1 \geq \dots \geq y_m$. 
We will show that if $y_1 > y_m$, 
\begin{equation}
\frac{\partial F}{\partial y_m} - \frac{\partial F}{\partial y_1} \leq 0.
\label{eq:mass2}
\end{equation} This implies that as long as $y_1 > y_m$, decreasing $y_1$ and increasing $y_m$ by the same amount will never increase $F$ while maintaining $y_1 + \dots + y_m =y_s$, so $F$ is minimized when $y_1 = \dots = y_m = \frac{y_s}{m}$. The same argument for $z_1, \dots, z_m$ will prove the theorem.

Let $Y := Y_1 + \dots + Y_m$ and $Z := Z_1 + \dots + Z_m$, where each $Y_i$ (resp. $Z_i$) is an independent Bernoulli random variable with $\E[Y_i] = y_i$ (resp. $\E[Z_i] = z_i$). 
To prove~\eqref{eq:mass2}, we first compute 
$\frac{\partial \E[\frac{1}{1 + \max(Y, Z)}]}{ \partial y_m} - \frac{\partial \E[\frac{1}{1 + \max(Y, Z)}]}{\partial y_1}$. Let $Y' := Y_2 + \dots + Y_{m - 1}$. 
We decompose $\E[\frac{1}{1 + \max(Y, Z)}]$ as follows. 
\begin{align*}
& \E[\frac{1}{1 + \max(Y, Z)}]   \\
=& \sum_{i = 0}^m \sum_{j = 0}^m \bigg( \Pr[Y = i] \cdot \Pr[Z = j] \cdot \frac{1}{1 + \max(i, j) } \bigg) \\ 
=& \sum_{i = 0}^m  \bigg( \Pr[Y' = i] \cdot \Pr[Z \leq  i] \big(  \frac{(1 - y_1)(1 - y_m)}{1 + i} 
+ \frac{y_1(1 - y_m) + (1 - y_1)y_m}{2 + i}
+ \frac{y_1 y_m}{3 + i}
\big) \bigg) \\
+& \sum_{i = 0}^m \bigg( \Pr[Y' = i] \cdot \Pr[Z = i + 1] \big(  \frac{1 - y_1 y_m}{2 + i}
+ \frac{y_1 y_m}{3 + i}
\big) \bigg) +  \sum_{i = 0}^m \Pr[Y' = i] \cdot \Pr[Z \geq i + 2] \cdot \frac{1}{3 + i}
\end{align*}
Therefore, the directional derivative can be written as 
\begin{align*}
& (\frac{\partial}{\partial y_m } - \frac{\partial }{ \partial y_1}) 
\E[\frac{1}{1 + \max(Y, Z)}] \\
=& (y_1 - y_m) \sum_{i = 0}^m  \bigg( \Pr[Y' = i] \cdot \Pr[Z \leq  i] \big(  \frac{1}{1 + i} 
- \frac{2}{2 + i} + \frac{1}{3 + i}
\big) \bigg) \\
+& (y_1 - y_m) \sum_{i = 0}^m \bigg( \Pr[Y' = i] \cdot \Pr[Z = i + 1] \big( - \frac{1}{2 + i} + \frac{1}{3 + i} \big) \bigg) \\
\leq& (y_1 - y_m) \sum_{i = 0}^m  \bigg( \Pr[Y' = i] \cdot \Pr[Z \leq  i] \big(  \frac{1}{1 + i} 
- \frac{2}{2 + i} + \frac{1}{3 + i}
\big) \bigg) \\
\leq& (y_1 - y_m) \sum_{i = 0}^m  \bigg( \Pr[Y' = i] \cdot \Pr[Z \leq  i] \big(  \frac{1}{1 + i} 
- \frac{2}{2 + i} + \frac{1}{3 + i}
\big) \bigg) \\
\leq& \frac{y_1 - y_m}{3},
\end{align*}
where the last inequality follows from the fact that 
\[
\big(  \frac{1}{1 + i} 
- \frac{2}{2 + i} + \frac{1}{3 + i}
\big) = \frac{2}{(1 + i)(2 + i)(3 + i)} \leq \frac{1}{3}.
\]
Finally,
\begin{align*}
& (\frac{\partial}{\partial y_m } - \frac{\partial }{ \partial y_1})  F  \\
=& (\frac{\partial}{\partial y_m } - \frac{\partial }{ \partial y_1}) 
(x_e \E[\frac{1}{1 + \max(Y, Z)}] + c x_e y_1^2 - c x_e^2 y_1 + c x_e y_m^2 - c x_e^2 y_m) \\
\leq& \frac{x_e(y_1 - y_m)}{3} - 2c x_e (y_1 - y_m) = 0.
\end{align*}
By taking $c = \frac{1}{6}$. 
\end{proof}

Therefore, for any $e \in E$, $\beta_e(x) \geq F(x_e, \frac{y_s}{n - 1}, \dots, \frac{y_s}{n - 1}, \frac{z_s}{n - 1}, \dots, \frac{z_s}{n - 1})$ for some $y_s \leq 1 - x_e$ and $z_s \leq 1 - x_e$. 
Let 
\begin{align*} G(x_e, y_s, z_s) := & F(x_e, \frac{y_s}{n - 1}, \dots, \frac{y_s}{n - 1}, \frac{z_s}{n - 1}, \dots, \frac{z_s}{n - 1}) \\
=&
x_e \E[\frac{1}{1 + \max(Y, Z)}] + (n - 1) c (x_e (\frac{y_s}{n - 1})^2  - x_e^2 (\frac{y_s}{n - 1}))  \\
&+ (n - 1) c (x_e (\frac{z_s}{n - 1})^2 - x_e^2 (\frac{z_s}{n - 1})) \\
=& x_e \E[\frac{1}{1 + \max(Y, Z)}] + c x_e y_s ((\frac{y_s}{n - 1})  - x_e) + c x_e z_s ((\frac{z_s}{n - 1})  - x_e) \\
\geq & 
x_e \E[\frac{1}{1 + \max(Y, Z)}] - 2 c x_e^2
\end{align*}
where $Y \sim \Bino(n - 1, \frac{y_s}{n - 1})$, $Z \sim \Bino(n - 1, \frac{z_s}{n - 1})$. 
Note that the final quantity is minimized when $y_s = z_s = 1 - x_e$. Finally, let 
\[
H_{n - 1}(x_e) := x_e \E[\frac{1}{1 + \max(Y, Z)}] - 2 c x_e^2,
\]
where $Y, Z \sim \Bino(n - 1, \frac{1 - x_e}{n - 1})$. 

\begin{lem}
For any $m \in \mathbb{N}$ and $x_e \in [0, 1]$, $H_m(x_e) \geq 0.476 x_e$. 
\label{lem:finish}
\end{lem}
\begin{proof}
    Since the binomial distribution is approximated by the Poisson distribution in the limit,
    we use this to ease the calculation.  Let $Y,Z \sim \textsf{Poisson}(1-x)$.
    Let $H(x) := x \E[ \frac{1}{1+\max(Y,Z)}] - x^2/3$ (we substitute $c=1/6$ into the earlier equation).
    In particular, we write the expectation in full to get
    \begin{align*}
        \E[\frac{1}{1+\max(Y,Z)}] &= \sum_{k=0}^{\infty} \sum_{j=0}^{\infty} \frac{1}{1+\max(j,k)} e^{-2(1-x)} \frac{(1-x)^{j+k}}{j! k!}  \\
         &= \frac{1}{e^{2(1-x)}}  \sum_{k=0}^{\infty} \Big( \sum_{j=0}^k \frac{1}{1+\max(j,k-j)} \frac{1}{j! (k-j)!} \Big) (1-x)^{k}
    \end{align*}
    Let $P_t(x)$ denote the above sum truncated at $k=t$. I.e.~ 
    \[P_t(x):= \frac{1}{e^{2(1-x)}} \sum_{k=0}^{t}\Big( \sum_{j=0}^k \frac{1}{1+\max(j,k-j)} \frac{1}{j! (k-j)!} \Big) (1-x)^{k}\] 
    This is a degree $t$-polynomial in $(1-x)$ with a normalizing factor of $e^{-2(1-x)}$ and 
    note that $\E[\frac{1}{1+\max(Y,Z)}] \geq P_t(x)$ for any $t\in \mathbb{N}$. 

    Truncating this polynomial with $t=15$, we can see that this has a minimum value of $0.476$ for all values of $x \in [0,1]$. 
    we can see that $ \E[\frac{1}{1+\max(Y,Z)}] - x/3 \geq P_15(x) - x/3 $. In the interval $x\in [0,1]$, this function
    achieves its minimum at $x=0$ achieving a minimum of $0.476$. 
\end{proof}

\section {Weighted Bipartite Graphs}
\label{sec:bipartite_weighted}
We prove Theorem~\ref{thm:bipartite_weighted} for weighted bipartite graphs.
As explained in Section~\ref{subsec:techniques}, it suffices to prove that for each $e = (u, v) \in E$, 
\[
\E_{G \sim \caldb_{n,w, x}}\bigg[ \frac{1}{\max(\deg_G(u), \deg_G(v))} | e \in G\bigg] \geq 0.4481.
\]

Fix $e = (u, v)$ and assume $V = \{ v, v_1, \dots, v_{n-1} \} \cup \{ u, u_1, \dots, u_{n-1} \}$. 
Let $Y = \deg_G(u) - 1$ and $Z = \deg_G(v) - 1$. Given $e \in G$, 
$Y$ and $Z$ can be represented as 
$Y = \sum_{i=1}^{n-1} Y_i$ and
$Z = \sum_{i=1}^{n-1} Z_i$, where $Y_i$ indicates where $(u, v_i) \in E(G)$ and $Z_i$ indicates where $(v, u_i) \in E(G)$. This construction ensures that 
\[
\E_G\bigg[ \frac{1}{\max(\deg_G(u), \deg_G(v))} | e \in G\bigg] = 
\E_{Y, Z}\bigg[ \frac{1}{1 + \max(Y, Z)}\bigg].
\]
Note that $Y_1, \dots, Y_{n - 1}, Z_1, \dots, Z_{n - 1}$ are mutually independent, and $\E[Y], \E[Z] \leq 1$. By monotonicity, assuming $\E[Y] = \E[Z] = 1$ never increases the lower bound. 
The following theorem shows that the worst case happens when one of $Y, Z$ is consistently $1$ and the other is drawn from $\Bino(n - 1, \frac{1}{n - 1})$. 

\begin{thm}
Let $Y = Y_1 + \dots + Y_m$ and $Z = Z_1 + \dots + Z_m$, 
where $Y_1, \dots, Y_m, Z_1, \dots, Z_m$ are mutually independent Bernoulli random variables with $\E[Y] = \E[Z] = 1$. Then, 
\[
\E\bigg[\frac{1}{1 + \max(Y, Z)}\bigg] \geq \E\bigg[\frac{1}{1 + Y_U}\bigg],
\]
where $Y_U$ is drawn from $\Bino(m, \frac{1}{m})$. 
\label{thm:oneside}
\end{thm}
\begin{proof}
We decompose $\E[\frac{1}{1 + \max(Y, Z)}]$ as follows.

\begin{align*}
\E\bigg[\frac{1}{1 + \max(Y, Z)}\bigg]   &= \sum_{i = 0}^m \sum_{j = 0}^m \bigg( \Pr[Y = i] \cdot \Pr[Z = j] \cdot \frac{1}{1 + \max(i, j) } \bigg) \\ 
&= \sum_{i = 0}^m \Pr[Y = i]  \bigg[ \big( \sum_{j = 0}^i \Pr[Z = j] \big) \cdot \frac{1}{1 + i} + \big( \sum_{j = i + 1}^m \Pr[Z = j] \cdot \frac{1}{1 + j} \big) \bigg]  \\
&= \sum_{i = 0}^m \Pr[Y = i] \cdot \frac{1}{1 + i} - \sum_{i = 0}^m \Pr[Y = i] \bigg[  \sum_{j = i+1}^m \Pr[Z = j]  \big( \frac{1}{1 + i} - \frac{1}{1 + j} \big)  \bigg] \\
&= \sum_{i = 0}^m \Pr[Y = i] \cdot \frac{1}{1 + i} - \sum_{j = 1}^m \Pr[Z = j] \bigg[  \sum_{i = 0}^{j - 1} \Pr[Y = i]  \big( \frac{1}{1 + i} - \frac{1}{1 + j} \big)  \bigg].
\end{align*}

Let $t_j := \sum_{i = 0}^{j - 1} \Pr[Y = i] \cdot \big( \frac{1}{1 + i} - \frac{1}{1 + j} \big)$. We prove the following facts about $t_j$'s. 
\begin{lem}
For all $j \geq 3$, $\frac{t_2}{2} \geq \frac{t_j}{j}$. 
\label{lem:gain}
\end{lem}
\begin{proof}
Fix $j \geq 3$. By the definition of $t_2$ and $t_j$, 
\begin{align*}
\noindent \frac{t_2}{2} - \frac{t_j}{j} 
& = \frac{1}{2} \bigg( \Pr[Y = 0] (1 - \frac{1}{3}) + \Pr[Y = 1](\frac{1}{2} - \frac{1}{3}) \bigg)
 \bigg) - \frac{1}{j} \bigg(  \sum_{i = 0}^{j - 1} \Pr[Y = i]  \big( \frac{1}{1 + i} - \frac{1}{1 + j} \big) \bigg) \\
& = \frac{1}{3}  \Pr[Y = 0] + \frac{1}{12} \Pr[Y = 1]  - \frac{1}{j} \bigg(  \sum_{i = 0}^{j - 1} \Pr[Y = i]  \big( \frac{1}{1 + i} - \frac{1}{1 + j} \big) \bigg) \\
&= (\frac{1}{3} - \frac{1}{1 + j}) \Pr[Y = 0] + (\frac{1}{12} - \frac{j - 1}{2j(j + 2)}) \Pr[Y = 1] \\
&\quad -  \frac{1}{j} \bigg( \sum_{i = 2}^{j - 1} \Pr[Y = i]  \big( \frac{1}{1 + i} - \frac{1}{1 + j} \big) \bigg) \\
&\geq \bigg( \frac{1}{3} - \frac{1}{1 + j} - \frac{1}{j} \sum_{i = 2}^{j - 1} \big( \frac{1}{1 + i} - \frac{1}{1 + j} \big) \bigg) \Pr[Y = 0] + (\frac{1}{12} - \frac{j - 1}{2j(j + 2)}) \Pr[Y = 1],
\end{align*}
where the inequality follows from $\Pr[Y = 0] \geq \Pr[Y = i]$ for $i \geq 2$. To prove $\frac{t_2}{2} - \frac{t_j}{j} \geq 0$, it suffices to prove that $\frac{1}{3} - \frac{1}{1 + j} - \frac{1}{j} \sum_{i = 2}^{j - 1} \big( \frac{1}{1 + i} - \frac{1}{1 + j} \big)  \geq 0$, and $\frac{1}{12} - \frac{j - 1}{2j(j + 2)} \geq 0$. It is easy to verify the latter for $j \geq 3$. The former can be proved as 
\begin{align*}
& \frac{1}{3} - \frac{1}{1 + j} - \frac{1}{j} \sum_{i = 2}^{j - 1} \big( \frac{1}{1 + i} - \frac{1}{1 + j} \big) \\ 
=& 
\frac{1}{3} + \frac{j - 2}{j(1+j)} - \big( \frac{1}{1 + j}  + \frac{1}{j} \sum_{i = 2}^{j - 1}  \frac{1}{1 + i} \big) \\
\geq & 
\frac{1}{3} + \frac{j - 2}{j(1+j)} - \big( \frac{1}{1 + j}  + \frac{j - 2}{3j} \big) \\
= & \big( \frac{1}{3} - \frac{j - 2}{3j} \big) + \big( \frac{j - 2}{j(1+j)} - \frac{1}{1 + j} \big) \\
= & \frac{2}{3j} - \frac{2}{j(1 + j)} \geq 0,
\end{align*}
where the first inequality follows from $\frac{1}{1+i} \leq \frac{1}{3}$ for $i \geq 2$ and the last inequality follows from $j \geq 3$. 
\end{proof}

We prove the theorem by considering the following two cases.
\paragraph*{Case 1: $2\Pr[Y=0] \geq \Pr[Y = 1]$ or $2\Pr[Z=0] \geq \Pr[Z = 1]$. } 
Without loss of generality, assume that $2\Pr[Y=0] \geq \Pr[Y = 1]$. It is equivalent to 
\begin{align*}
& \Pr[Y = 0] \geq \frac{2}{3} \Pr[Y = 0] + \frac{1}{6} \Pr[Y = 1] \\
\Leftrightarrow \quad & t_1 \geq \frac{t_2}{2}.
\end{align*}
By Lemma~\ref{lem:gain}, it implies that $t_1 \geq \frac{t_j}{j}$ for all $j \geq 2$. Then, since $\E[Z] = \sum_{j = 1}^m j \cdot \Pr[Z = j] = 1$, 
\begin{align*}
\E\bigg[\frac{1}{1 + \max(Y, Z)}\bigg]
&= \sum_{i = 0}^m \Pr[Y = i] \cdot \frac{1}{1 + i} - \sum_{j = 1}^m \Pr[Z = j] t_j \\
&\geq \sum_{i = 0}^m \Pr[Y = i] \cdot \frac{1}{1 + i} - t_1 \sum_{j = 1}^m j \cdot \Pr[Z = j]  \\
&= \sum_{i = 0}^m \Pr[Y = i] \cdot \frac{1}{1 + i} - t_1 \\
&= \E[\frac{1}{1 + \max(Y, 1)}].
\end{align*}

The following lemma proves the theorem in the case $t_1 \geq \frac{t_2}{2}$. 
\begin{lem}
$\E[\frac{1}{1 + \max(Y, 1)}] \geq \E[\frac{1}{1 + \max(Y_U, 1)}]$. 
\label{lem:one_side}
\end{lem}
\begin{proof}
Note that $Y = Y_1 + \dots + Y_m$, and each $Y_i$ is a Bernoulli random variable. Let $y_i := \E[Y_i]$. Without loss of generality, assume $y_1 \geq \dots \geq y_m$. We will show that if $y_1 > y_m$, 
\begin{equation}
\frac{\partial \E[\frac{1}{1 + \max(Y, 1)}]}{\partial y_m} - \frac{\partial \E[\frac{1}{1 + \max(Y, 1)}]}{\partial y_1} \leq 0.
\label{eq:mass1}
\end{equation} This implies that as long as $y_1 > y_m$, decreasing $y_1$ and increasing $y_m$ by the same amount will never increase $\E[\frac{1}{1 + \max(Y, 1)}]$ while maintaining $y_1 + \dots + y_m =1$, so the expectation is minimized when $y_1 = \dots = y_m$, or $Y = Y_U$. Consider the following decomposition of $\E[\frac{1}{1 + \max(X, Y)}]$. 
\begin{align*}
\E_Y\bigg[\frac{1}{1 + \max(1, Y)}\bigg] &= 
\Pr[Y = 0]\cdot \frac{1}{2} + \sum_{i = 1}^m \Pr[Y = i]\cdot \frac{1}{1 + i} \\
&= \frac{1}{2} (1 - \sum_{i = 1}^m \Pr[Y = i])  + \sum_{i = 1}^m \Pr[Y = i]\cdot \frac{1}{1 + i} \\
&= \frac{1}{2} - \sum_{i = 2}^m \Pr[Y = i]\cdot (\frac{1}{2} - \frac{1}{1 + i}) \\
&= \frac{1}{2} - \sum_{i = 2}^m \Pr[Y \geq i]\cdot (\frac{1}{i} - \frac{1}{1 + i}).
\end{align*}
To prove~\eqref{eq:mass1}, it suffices to prove that for all $i \geq 2$, 
\begin{equation*}
\frac{\partial \Pr[Y \geq i]}{\partial y_m} - \frac{\partial \Pr[Y \geq i]}{\partial y_1} \geq 0.
\end{equation*}
Let $Y' = Y_2 + \dots + Y_{m - 1}$, and fix $i \geq 3$. 
\begin{align*}
\Pr[Y \geq i] &= \Pr[Y' = i - 2] y_1 y_m + \Pr[Y' = i - 1] \big( y_1(1 - y_m) + (1 - y_1)y_m + y_1 y_m) \\
 &\quad + \Pr[Y' \geq i] \\
\frac{\partial \Pr[Y \geq i]}{\partial y_1} &= \Pr[Y' = i - 2] y_m + \Pr[Y' = i - 1] \big( 1 - y_m \big)
\end{align*}
Therefore, 
\begin{align*}
\frac{\partial \Pr[Y \geq i]}{\partial y_m} - 
\frac{\partial \Pr[Y \geq i]}{\partial y_1} &= \Pr[Y' = i - 2] (y_1 - y_m) + \Pr[Y' = i - 1] \big( y_m - y_1 \big) \\
&= (y_1 - y_m) \big( \Pr[Y' = i - 2]  + \Pr[Y' = i - 1] \big).
\end{align*}
Finally, it remains to show that $\Pr[Y' = j] \geq \Pr[Y' = j + 1]$ for all $j \geq 0$. The case $j = 0$ is true since $\Pr[Y' = 0] = \prod_{k = 2}^{m - 1} (1 - y_k)$ and 
\[
\Pr[Y' = 1] = \sum_{k = 2}^{m - 1} \Pr[Y' = 0] \cdot \frac{y_k}{1 - y_k} 
\leq \sum_{k = 2}^{m - 1} \Pr[Y' = 0] \frac{y_k}{1 - y_2} 
= \frac{\Pr[Y' = 0]}{1 - y_2} \sum_{i=2}^{m-1} y_k \leq \Pr[Y' = 0],
\]
where the last line follows from $\sum_{k=2}^{m - 1} y_i \leq 1 - y_1 \leq 1 - y_2$ since $y_1$ is the biggest element. The case $j \geq 1$ follows from the fact the sequence $( \Pr[Y' = j] )_j$ has one mode or two consecutive modes, and at least one of them occurs at $j = 0$ ($\E[Y'] < 1$ implies $\Pr[Y' = 0] > \Pr[Y' = j]$ for all $j \geq 2$).
\end{proof}

\paragraph*{Case 2: $2\Pr[Y=0] \leq \Pr[Y = 1]$ and $2\Pr[Z=0] \leq \Pr[Z = 1]$.} 
Since $\sum_{i=0}^m \Pr[Z = i] = 1$ and $\E[Z] = \sum_{i=1}^m i \cdot \Pr[Z = i] = 1$, we have $\Pr[Z = 0] = \sum_{i = 2}^m (i - 1) \Pr[Z = i]$. Together with the fact $2\Pr[Z = 0] \leq \Pr[Z = 1]$, it implies 
\begin{align*}
& 1 - \Pr[Z = 1] = \Pr[Z = 0] + \sum_{i=2}^m \Pr[Z = i] 
\leq 2 \Pr[Z = 0] < \Pr[Z = 1],
\end{align*}
so $\Pr[Z = 1] \geq \frac{1}{2}$. Finally,
\begin{align*}
\E\bigg[\frac{1}{1 + \max(Y, Z)}\bigg]
&= \sum_{i = 0}^m \Pr[Y = i] \cdot \frac{1}{1 + i} - \sum_{j = 1}^m \Pr[Z = j]\cdot t_j \\
&= \sum_{i = 0}^m \Pr[Y = i] \cdot \frac{1}{1 + i} - \Pr[Z = 1]\cdot t_1 - \sum_{j = 2}^m \Pr[Z = j] \cdot t_j \\
&\geq \sum_{i = 0}^m \Pr[Y = i] \cdot \frac{1}{1 + i} - \Pr[Z = 1]\cdot t_1 - \sum_{j = 2}^m j \cdot \Pr[Z = j] \cdot \frac{t_2}{2} \\
&= \sum_{i = 0}^m \Pr[Y = i] \cdot \frac{1}{1 + i} - \Pr[Z = 1]\cdot t_1 - \frac{t_2}{2} (1 - \Pr[Z = 1]) \\
&\geq \sum_{i = 0}^m \Pr[Y = i] \cdot \frac{1}{1 + i} - \frac{t_1}{2} - \frac{t_2}{4}
= \E\bigg[\frac{1}{1 + \max(Y, Y_{H})}\bigg], 
\end{align*}
where $Y_H$ is drawn from $\Bino(2, \frac{1}{2})$.
The first inequality follows from Lemma~\ref{lem:gain}, 
and the second inequality follows from $\Pr[Z = 1] \geq 0.5$ and $t_1 \leq \frac{t_2}{2}$. 

Since $Y_{H}$ satisfies $2\Pr[Y_{H} = 0] = \Pr[Y_{H} = 1]$, the analysis for Case 1 shows that $\E[\frac{1}{1 + \max(Y, Y_{H})}] \geq \E[\frac{1}{1 + \max(1, Y_U)}]$. 
\end{proof}
The following lemma finishes the proof of Theorem~\ref{thm:bipartite_weighted}.
\begin{lem}
For any $m \in \mathbb{N}$, if $Y \sim \Bino(m, \frac{1}{m})$, 
\[
\E\bigg[\frac{1}{1 + \max(1, Y)}\bigg] \geq 0.4481
\]
\end{lem}
\begin{proof}
    Since the binomial distribution is approximated by the Poisson distribution in the limit,
    we use this to ease the calculation.  Let $Y \sim \textsf{Poisson}(1)$.
    \begin{align*}
        \E\bigg[\frac{1}{1+\max(1,Y)}\bigg] &= \sum_{k=2}^{\infty}  \frac{1}{k+1} \Pr[Y'=k] + \frac{1}{2}\Pr[Y<2] \\
        &= \sum_{k=2}^{\infty} \frac{1}{k+1} \frac{1}{k! \cdot e}  + \frac{1}{2} (\frac{1}{e} + \frac{1}{e}) \\
        &=  \frac{1}{e} \big(\sum_{k=0}^{\infty} \frac{1}{k!} -1-1-\frac{1}{2} \big)  + \frac{1}{2} (\frac{1}{e} + \frac{1}{e}) \\
        &= (e-\frac{5}{2})\frac{1}{e}  + \frac{1}{e}  \\
        &\geq 0.4481
    \end{align*}
\end{proof}

\section {General Graphs}
\label{sec:general_weighted}
In this section, we prove Theorem~\ref{thm:general_weighted} for weighted general graphs. Our proof methods here closely follow that of Lemma 4.9 of Chekuri et al.~\cite{CVZ14} that lower bounds the correlation gap for monotone submodular functions by $1 - 1/e$. 
The only difference is that Lemma~\ref{lem:continuous} holds for matching with a weaker guarantee (if $\nu$ was a monotone submodular function, Lemma~\ref{lem:continuous} would hold with $2\nu(G)$ replaced by $\nu(G)$).

\begin{proof}
Fix weights $w \in (\RR^+ \cup \{ 0 \})^E$. 
Define $F : [0, 1] \to (\RR^+ \cup \{ 0 \})$ as
$F(x) := \E_{G \sim \caldg_{n,w, x}} [\nu(G)]$. 
Now, fix $x \in [0, 1]^E$ in the matching polytope. 
We will show $F(x) \geq 0.43 \sum_{e \in E} w_e x_e$. 

Consider the function $\phi(t) := F(tx)$ for $t \in [0, 1]$.
\begin{equation}
\frac{d \phi}{dt} = x \cdot \nabla F (tx) = \sum_{e \in E} x_e \frac{\partial F}{\partial x_e} \bigg|_{tx}
\label{eq:phi}
\end{equation}
For each $e \in E$, 
\begin{align*}
\frac{\partial F}{\partial x_e} \bigg|_{tx} 
&= \frac{\partial \E_{G \sim \caldg_{n,w, tx}} [\nu(G)]}{\partial x_e} \bigg|_{tx} \\
&= \E_{G \sim \caldg_{n,w, tx}} [\nu(G) | e \in G] - \E_{G \sim \caldg_{n,w, tx}} [\nu(G) | e \notin G] \\
&= \E_{G \sim \caldg_{n,w, tx}} [\nu(G \cup \{ e \})  -\nu(G \setminus \{ e \})],
\end{align*}
where $G \cup \{ e \}$ (resp. $G \setminus \{ e \}$) denotes the graph $(V, E(G) \cup \{ e \})$ (resp. $(V, E(G) \setminus \{ e \})$. 
\begin{lem}
For any fixed graph $G$ with weights $\{ w_e \}$ and any point $x$ in the matching polytope, 
\[
\sum_{e \in E} x_e \big( \nu(G \cup \{ e \})  -\nu(G \setminus \{ e \}) \big) + 2 \nu(G) \geq \sum_{e \in E} x_e w_e.
\]
\label{lem:continuous}
\end{lem}
\begin{proof}
Let $M \subseteq E(G)$ be a maximum weight matching of $G$. 
Note that 
\begin{align}
& \sum_{e \in E} x_e \big( \nu(G \cup \{ e \})  -\nu(G \setminus \{ e \}) \big) + 2 \nu(G) \nonumber \\
\geq& \sum_{e \in E} x_e \big( \nu(G \cup \{ e \})  -\nu(G) \big) + 2 \sum_{f \in M} w_f \nonumber \\
\geq& \sum_{e \in E} x_e \big( \nu(G \cup \{ e \})  -\nu(G) \big) + \sum_{f \in M} \sum_{e \in E: e \sim f} x_e w_f
\label{eq:continuous}
\end{align}
where $f \sim e$ indicates that two edges $f$ and $e$ share an endpoint.
To prove the lemma, it suffices to show that for each $e \in E$, the coefficient of of $x_e$ in~\eqref{eq:continuous} is at least $w_e$. We consider the following cases.

\begin{itemize}
\item If $M \cup \{ e \}$ is a matching, $\nu(G \cup \{ e \}) \geq \nu(G) + w_e$ and $\nu(G \setminus \{ e \}) \leq \nu(G)$, so $\nu(G \cup \{ e \})  -\nu(G \setminus \{ e \}) \geq w_e$. 

\item If $e$ intersects exactly one edge $f \in M$, the coefficient of $x_e$ is 
$\nu(G \cup \{ e \})  -\nu(G) + w_f$. If $w_f \geq w_e$, it is at least $w_e$. If $w_f < w_e$, $M \cup \{ e \} \setminus \{ f \}$ is a matching of weight $\nu(G) + w_e - w_f$. 
It implies that $e \notin E(G)$ and $\nu(G \cup \{ e \}) - \nu(G) \geq w_e - w_f$, so 
$\nu(G \cup \{ e \})  -\nu(G) + w_f \geq w_e$. 

\item If $e$ intersects two edges $f, g \in M$, the coefficient of $x_e$ is 
$\nu(G \cup \{ e \})  -\nu(G) + w_f + w_g$. If $w_f + w_g \geq w_e$, it is at least $w_e$. If $w_f + w_g < w_e$, $M \cup \{ e \} \setminus \{ f, g \}$ is a matching of weight $\nu(G) + w_e - w_f - w_g$. 
It implies that $e \notin E(G)$ and $\nu(G \cup \{ e \}) - \nu(G) \geq w_e - w_f - w_g$, so 
$\nu(G \cup \{ e \})  -\nu(G) + w_f + w_g \geq w_e$. 
\end{itemize}
\end{proof}
Combining~\eqref{eq:phi} and Lemma~\ref{lem:continuous}, 
\begin{align*}
\frac{d \phi}{dt} &= \sum_{e \in E} x_e \frac{\partial F}{\partial x_e} \bigg|_{tx} \\
&=\sum_{e \in E} \E_{G \sim \caldg_{n,w, tx}}[\nu(G \cup e) - \nu(G \setminus e)] \\
&\geq \sum_{e \in E} x_e w_e - 2 \E_{G \sim \caldg_{n,w, tx}}[\nu(G)] \\
&= \sum_{e \in E} x_e w_e - 2 \phi(t).
\end{align*}
which implies that, 
\[
\frac{d}{dt} (e^{2t} \phi(t)) = 2 e^{2t} \phi(t) + e^{2t} \frac{d\phi}{dt} \geq e^{2t} \sum_{e \in E} x_e w_e.
\]
Since $\phi(0) = 0$, 
\[
e^2 \phi(1) \geq \sum_{e \in E} x_e w_e \int_{0}^1 e^{2t} dt = \frac{e^2 - 1}{2} \sum_{e \in E} x_e w_e,
\]
which proves the theorem.
\end{proof}

\bibliographystyle{abbrv}
\bibliography{mybib-randommatching}

\end{document}